\documentclass[10pt, conference]{ieeeconf}
\IEEEoverridecommandlockouts
\usepackage{cite}
\usepackage{amsmath,amssymb,amsfonts}
\usepackage{algorithmic}
\usepackage{graphicx}
\usepackage{blindtext}
\usepackage{hyperref}
\usepackage[english]{babel}

\usepackage{amsthm}
\usepackage{todonotes}
\setuptodonotes{inline}
\usepackage{xcolor}

\usepackage{background}
\SetBgContents{Copyright (c) 2026 IEEE}
\SetBgScale{1}
\SetBgAngle{0}
\SetBgPosition{current page.north east}
\SetBgHshift{-2.5cm}
\SetBgVshift{-1cm}

\newtheorem{theorem}{Theorem}[section]
\newtheorem{corollary}{Corollary}[theorem]

\newtheorem{definition}[theorem]{Definition}
\newtheorem{proposition}[theorem]{Proposition}

\usepackage{url}

\input{macros.sty}

\newtheorem{thm}{Theorem}[section]

\newtheorem{defn}{Definition}[section]

\newtheorem{rem}{Remark}

\newtheorem{assumption}{Assumption}

\usepackage{textcomp}
\usepackage{xcolor}
\def\BibTeX{{\rm B\kern-.05em{\sc i\kern-.025em b}\kern-.08em
    T\kern-.1667em\lower.7ex\hbox{E}\kern-.125emX}}

\begin{document}

\title{\LARGE \bf A Time-to-Collision Barrier Function Approach to Collision Avoidance for Stochastic Systems

\author{Benedikt Barthel Sorensen$^1$ \and Mitchell Black$^2$ \and Erfaun Noorani$^2$ \and Themistoklis P. Sapsis$^1$}
\thanks{DISTRIBUTION STATEMENT A. Approved for public release. Distribution is unlimited.}
\thanks{Accepted for presentation at the 65th IEEE Conference on Decision and Control (CDC 2026), Honolulu, Hawaii, USA.}
\thanks{$^1$Department of Mechanical Engineering, Massachusetts Institute of Technology, Cambridge, MA, USA; \texttt{$\{$bbarthel,sapsis$\}$@mit.edu}.}
\thanks{$^2$MIT Lincoln Laboratory, Lexington, MA, USA; \texttt{first.last@ll.mit.edu}.}
}

\maketitle

% \listoftodos

\newpage

\begin{abstract}
Collision avoidance constraints for autonomous systems are typically formulated in position or velocity space, implicitly reacting to geometric proximity. We propose an alternative paradigm based on the adversarial time-to-collision (aTTC): the minimum time in which an adversary could achieve a collision given its dynamical constraints. By defining a control barrier function (CBF) directly in the time domain, the resulting controller is inherently anticipatory. The evading agent responds not only to whether a pursuer is on a collision course, but to how quickly it could reach one. This formulation enables velocity modulation that exploits the pursuers dynamic limits as an evasive strategy, a behavior not captured by standard distance-based CBFs. Since exact aTTC computation requires integrating the full system dynamics, we employ a lightweight neural network surrogate that admits a real-time quadratic program-based control law. We validate the approach in a 2D comparative study and a 3D multi-agent pursuit-evasion scenario, where the aTTC-based CBF outperforms a higher-order distance-based baseline by more effectively buying time against superior pursuers with a significant speed advantage.
\end{abstract}
\section{Introduction} 
%% Paragraph 1 - Multi-Agent Collision Avodiance
Collision avoidance in multi-agent autonomous systems is a fundamental problem with broad relevance, such as robotic navigation in complex environments to autonomous vehicles operating in dense urban traffic. It requires tracking other agents' positions, velocities, and behaviors while minimizing collision risk. This is challenging due to dynamic, uncertain interactions, demanding efficient, real-time risk assessment that can be integrated into control policies.

% \subsection{Related Work.}
One approach is the control barrier function (CBF), which enforces safety through a forward invariance condition that can be encoded as a linear constraint on the control solution to a quadratic program (QP). The domain of safety is defined with respect to the barrier function. Forward invariance, and thus safety, is then obtained by placing a constraint on the rate-of-change of the barrier function, ensuring that no barrier crossings will occur~\cite{ames_control_2017,ames_control_2019,wang2017safety, zhang2025gcbf+}. Another approach is the velocity obstacle (VO), which defines for a given agent and a moving obstacle the set of velocity vectors which would give rise to a collision if both maintained constant velocity for some time~\cite{fiorini_motion_1998}. These ideas have recently been combined in so called Velocity Obstacle CBFs (VO-CBFs) which construct the barrier function in terms of the VO geometry rather than standard distance-based formulation \cite{huang2025dynamic,roncero2025multi}. 
% Working directly in velocity space streamlines the CBF formulation by allowing a first order definition while embedding the VO in the CBF-QP framework allows for less conservative interventions. 
Other CBF alternatives have been proposed including the Future-Focused CBF~\cite{black2023future}, which guarantees safety over arbitrarily long time horizons, and the Collision Cone CBF (C3BF), which is based on line-of-sight to the obstacle \cite{tayal2026collision}. 

%% Paragraph 2 - Adverserial Collision Avoidance
In many situations the obstacles encountered by an agent are not simply neutral moving objects, but other autonomous agents. If these agents are \textit{cooperative}, the burden of avoiding collision between any pairs of agents can be split such that each agent only executes a fraction of the necessary evasive maneuver \cite{van_den_berg_reciprocal_2008}. Such a control scheme can be decentralized and parallelized as long as each agent \textit{assumes} that all others obey the same control logic \cite{van_den_berg_reciprocal_2011}. \textit{Adversarial} collision avoidance is the opposite limiting case where the agent(s) must assume that any other agent is \textit{actively seeking} collision. This problem has been addressed from game-\cite{exarchos_uav_2016} and control-theoretic perspectives such as CBFs \cite{lv_control_2024}. However, these studies were restricted to 2D motion with fixed speed and considered only very short duration missions and thus did not provide statistical analysis of the long term collision avoidance capability. 

%% Paragraph 3 - TTC and Our Method
In this work, we propose a temporal CBF formulation where the barrier function is neither defined in Euclidean space nor velocity, but directly in time. To this end, we build on the concept of time-to-collision (TTC), a standard metric in automotive safety \cite{hayward1972near,barhoumi_formal_2026}. At any instant the TTC is the time to collision between two agents were they to continue on their current trajectory at their current speed. Time is in fact the natural dimension in which to measure collision risk. Consider that when negotiating rush-hour traffic, ski slopes, or a crowded concert venue, we generally navigate by an intuitive sense of our reaction time rather than an explicit evaluation of relative distances and speeds. TTC also captures both distance and speed in one measure: something far away but fast can be more dangerous than something close but still. In fact, this concept has been applied to collision avoidance with static obstacles \cite{bosnak2017efficient}, and integrated with the previously mentioned VO-CBF framework \cite{roncero2025multi}.

\noindent\textbf{Contributions.}
The main contributions of this work are:
\begin{itemize}
    \item \textbf{Time-to-collision metric:} We introduce a \textit{worst-case} extension of TTC called adversarial TTC (aTTC) where it is assumed that one agent drives towards collision under maximum control effort. Assuming controllable dynamics, this metric is always finite and thus serves as a natural fit for real-time control. %in 3D environments. 
    \item \textbf{TTC-CBF:} We demonstrate how to construct a stochastic CBF using aTTC, and show that its use in control design confers a lower bound on the probability that a multi-agent system remains collision-free. 
    \item \textbf{Experimental validation:}  We train a differentiable surrogate model for aTTC, synthesize it in a stochastic CBF-QP control law, and compare this approach against existing methods in proactive CBF control on a set of planar collision avoidance trials. We also highlight its efficacy on a 3D pursuit-evasion scenario.
\end{itemize}

The paper is organized as follows. In Section \ref{sec:prob}, we review preliminaries and formalize the problem under consideration. Section \ref{sec:ttc-cbf} introduces our solution, the TTC-CBF. In Section \ref{sec:results}, we demonstrate the efficacy of our method over existing approaches in two simulated case studies, and in Section \ref{sec:conclusion}, we conclude with a summary and directions for future work.

\section{Problem Formulation}\label{sec:prob}
\textbf{Notation}: $\R$ and $\R_+$ denote the set of real and non-negative real numbers respectively.
A bolded $\state_t$ denotes a vector stochastic process at time $t$.
The trace of a matrix $\bb{M} \in \R^{n \times n}$ is $\textrm{Tr}(\bb{M})$, and $\blkdiag(\bb{M}_1, \dots, \bb{M}_n)$ denotes the block-diagonal matrix with diagonal blocks $\bb{M}_1, \dots, \bb{M}_n$. 
The Lie derivative of a function $\psi:\mathbb R^n\rightarrow \mathbb R$ along a vector field $f:\mathbb R^n\rightarrow\mathbb R^n$ at a point $x\in \mathbb R^n$ is $L_f\psi(x) \triangleq \frac{\partial \psi}{\partial x} f(x)$.

Consider a system of $N$ agents, where each agent $i \in \mathcal{I}$ evolves according to the stochastic differential equation
\begin{equation}\label{eq:agent-dynamics}
    \mathrm{d}\state_t^i = \mu^i(\state_t^i, \control_t^i)\,\mathrm{d}t + \sigma^i(\state_t^i)\,\mathrm{d}\W_t^i, \quad \state_0^i \in \Xset_0^i,
\end{equation}
where $\state_t^i \in \Xset^i \subseteq \R^n$ denotes the state and $\Xset_0^i$ the compact initial set, the control $\control_t^i = k^i(\state_t^1,\hdots,\state_t^N): \Xset^1 \times \hdots \times \Xset^N \to \Uset^i \subseteq \R^m$ is generated by a memoryless, state-feedback control law, and $\W_t^i$ is a standard $q$-dimensional Wiener process on a complete probability space. The drift is control-affine,
\begin{equation}\label{eq:drift}
    \mu^i(\state, \control) \triangleq f^i(\state) + g^i(\state)\control,
\end{equation}
where $f^i: \Xset^i \to \R^n$, $g^i: \Xset^i \to \R^{n \times m}$, and diffusion term $\sigma^i: \Xset^i \to \R^{n \times q}$ are all locally Lipschitz.
We let $\Phi^{k^i}(\state^i, \tau)$ denote the flow of the deterministic closed-loop dynamics of \eqref{eq:agent-dynamics} under control policy $k^i$, mapping initial condition $\state^i$ at time $t$ to the state at time $t + \tau$.
Stacking the individual states, inputs, and disturbances yields the joint system
\begin{equation}\label{eq:joint-dynamics}
    \mathrm{d}\state_t = \big(F(\state_t) + G(\state_t)\control_t\big) + \Sigma(\state_t)\mathrm{d}\W_t, \quad \state_0 \in \Xset_0,
\end{equation}
where $\state = (\state^1,\dots,\state^{N}) \in \Xset$, $\control = k(\state) = (\control^1,\dots,\control^{N}) \in \mathcal{U}$, $\W = (\W^1,\dots,\W^{N})$, and $F(\state) = [f^1(\state^1)^\top, \dots, f^N(\state^N)^\top]^\top$, and $G(\state) = \blkdiag\big(g^1(\state^1), \dots, g^N(\state^N)\big)$.

Over any compact domain $\mathcal{D} \subset \Xset$, the above conditions guarantee existence of a strong solution to~\eqref{eq:joint-dynamics}. For strong solutions, the (infinitesimal) generator is defined as follows.
\begin{definition}\hspace{-0.2pt}\cite[Def. 7.3.1]{Oksendal2003Stochastic}\label{def.generator}
    The (infinitesimal) generator $\mathcal{A}$ of $\state_t$ is defined by
    \begin{equation*}%\label{eq:generator}
        \mathcal{A}\psi(\bb{y}) = \lim_{t \downarrow 0}\frac{\mathbb{E}\left[\psi(\state_t) \; | \; \state_0 = \bb{y}\right] - \psi(\bb{y})}{t},
    \end{equation*}
    where $\psi: \R^n \mapsto \R$ belongs to  $\mathcal{D}_\mathcal{A}$, the set of all functions such that the limit exists for all $\state \in \R^n$.
\end{definition}
Analogous to the Lie derivative for deterministic systems, the generator characterizes the expectation of the derivative of a function $\psi$ over the trajectories of \eqref{eq:joint-dynamics}. By \cite[Thm. 7.3.3]{Oksendal2003Stochastic}, for a twice continuously differentiable function $\psi$ with compact support, the generator $\mathcal{A}$ of $\state_t$ is described by
\begin{equation}
    \scalebox{0.9}{$\displaystyle\mathcal{A}\psi(\state) = L_F \psi(\state) + L_G \psi(\state)k(\state) + \frac{1}{2}\textrm{Tr}\left(\Sigma(\state)^T\frac{\partial^2 \psi}{\partial \state^2}\Sigma(\state)\right), \nonumber $}
\end{equation}
alternatively denoted $\mathcal{A}\psi(\state, \control)$ by replacing $k(\state)$ with $\control$.

The objective of this paper is to design a memoryless, state-feedback controller $\control^i_t = k^i(\state_t)$ for each agent $i \in \mathcal{I}$ such that the probability of any inter-agent collision over $[0,T]$ is bounded from above.
The predominant approach in the literature is to define collision-free states as those for which inter-agent spatial distances remain above a safe threshold. 
We instead use temporal distance: the predicted time remaining until a spatial collision, and thus define the \emph{time-to-collision} (TTC) between agents $i$ and $j$ as
\begin{equation}\label{eq:time-to-collision}
\scalebox{0.95}{$\displaystyle \tau_{ij}^*(\state) \triangleq \inf\left\{\tau \geq 0 \,\middle|\, \|\Phi^{k^i}(\state^i, \tau) - \Phi^{k^j}(\state^j, \tau)\| \leq 2r\right\},$}
\end{equation}
where 
$r > 0$ is the collision radius.
We also introduce a critical time $\tau_c > 0$ as a lower bound on the admissible inter-agent TTCs so that all states $\state \in \Xset$ for which $\tau_{ij}^*(\state) - \tau_c \geq 0$ are deemed safe. 
Practically, $\tau_c$ quantifies the minimum time available for an intervening policy (e.g., emergency braking) to respond before a collision occurs.
Stochastic control barrier functions (SCBFs) are one tool for limiting the probability of exiting such a set of safe states.

% \textbf{Objective:} Design control inputs $u_i(t)$ for all agents to ensure collision avoidance over a finite time horizon $t \in [0,T]$.
Let $B: \Xset \to \R_+$ be a twice continuously differentiable function encoding the unsafe (colliding) and safe (collision-free) configurations as $\mathcal{C}$ and $\mathcal{S}$ respectively, such that
\begin{align}
    \mathcal{C} &= \{x \in \Xset \mid B(x) \geq 1\}, \label{eq:collision-set} \\
    \mathcal{S} &= \{x \in \Xset \mid 0 \leq B(x) < 1\}. \label{eq:safe-set}
\end{align}
It is assumed that $B(\state_0) \leq \gamma < 1$.
Let $\rho$ denote the probability that the system enters the unsafe set within a time interval of length $T <\infty$, i.e.,
\begin{equation}\label{eq:rho}
    \rho \triangleq \mathbb{P}\left\{\exists t \in [0, T]: \state \in \mathcal{C} \mid B(\state_0) \leq \gamma\right\}.
\end{equation}
% Stochastic control barrier functions (SCBFs) are one tool to limit the risk of a system becoming unsafe.
\begin{definition}[\hspace{-0.3pt}\cite{Yaghoubi2021RiskBounded}]\label{def:scbf}
    Given the set $\mathcal{S} \subset \mathcal{D} \subset \Xset$ defined by \eqref{eq:safe-set} for a twice continuously differentiable function $B\colon \Xset \to \R_+$, the function $B$ is a \textbf{stochastic control barrier function} (SCBF) for \eqref{eq:joint-dynamics} on $\mathcal{D}$ if there exist $\alpha, \beta \geq 0$ such that, for all $x \in \mathcal{D}$,
    \begin{equation}\label{eq:cbf-condition}
        \inf_{\control \in \Uset}\mathcal{A} B(\state, \control) \leq -\alpha B(\state) + \beta.
    \end{equation}
\end{definition}
\begin{theorem}[\hspace{-0.3pt}\cite{Yaghoubi2021RiskBounded}]\label{thm:scbf-safety}
    Consider a stochastic system of the form \eqref{eq:joint-dynamics} and a set $\Sset$ defined by a function $B$ as in \eqref{eq:safe-set}. If $B$ is a SCBF for \eqref{eq:joint-dynamics} over $\Sset$, then
    \begin{equation}\label{eq:scbf-safety-probability}
        \rho \leq \begin{cases}
        1 - \left(1 - \gamma\right)e^{-\beta T}; & \alpha > 0 \; \textrm{and} \; \alpha \geq \beta, \\
        \left(\gamma + (e^{\beta T} - 1)\frac{\beta}{\alpha}\right)e^{-\beta T}; & \alpha > 0 \; \textrm{and} \; \alpha < \beta, \\
        \gamma + \beta T; & \alpha=0.
        \end{cases}
    \end{equation}
\end{theorem}
Because the joint dynamics~\eqref{eq:joint-dynamics} inherit a block-diagonal structure from the decoupled agent dynamics~\eqref{eq:agent-dynamics}, the Lie derivatives present in $\mathcal{A}B(\state)$ in~\eqref{eq:cbf-condition} decompose as
\begin{equation*}
    L_F h(x) = \sum_{i=1}^{N} \frac{\partial h}{\partial x_i} f_i(x_i), \quad
    L_G h(x)u = \sum_{i=1}^{N} \frac{\partial h}{\partial x_i} g_i(x_i)\,u_i,
\end{equation*}
so that the SCBF condition can be evaluated and enforced using each agent's local dynamics, while the safety guarantee remains rigorously defined at the joint level.

% While the collision set $\mathcal{C}$, and by extension the function $B$, can be defined in many ways, for example with distance metrics like a static collision radius or future proximity under linearly predicted trajectories \cite{black2023future}, or with velocity-based metrics like velocity obstacles \cite{huang2025dynamic} or collision cones \cite{tayal2026collision}, a time-to-collision provides a more intuitive and practical notion of safety.

In Section \ref{sec:ttc-cbf}, we introduce the notion of TTC-CBFs for collision-free multi-agent control design.

\section{Time-to-Collision CBFs}\label{sec:ttc-cbf}
% To address the worst-case collision avoidance problem described above, we reinterpret the interactions with other agents and disturbances as an aggregate uncertainty acting on the controlled agent. 
% This allows the safety objective to be formulated as a robust control problem, where guarantees are required to hold for all admissible realizations of this uncertainty. 
% As an example of a worst-case collision scenario we consider an adversarial evader-pursuer configuration, where the \textit{evader} agent ($i=1$) attempts to avoid collision with a group of pursuing agents ($i>1$). 
% We focus on a single evader, but the following framework can be seamlessly extended to any number of evaders and pursuers.

In this section, we first introduce our time-to-collision (TTC) metric, which we use to formulate our Stochastic TTC-CBF. We then define the adversarial time-to-collision (aTTC), which quantifies the earliest time a pursuing agent could collide with a non-reactive agent. We construct a fast, fully differentiable surrogate model for this metric to enable efficient computation, and incorporate the aTTC model into a predictive control law, providing probabilistic collision avoidance guarantees.

\subsection{Time-to-Collision}
% We now define the following metric for \emph{time-to-collision} (TTC) between agents $i$ and $j$,
% \begin{equation}\label{eq:time-to-collision}
%     \tau_{ij}^*(\state)  = \inf\big\{\tau \geq 0 \mid \|\Phi(\state^i, \tau) - \Phi(\state^j, \tau)\| \leq 2r\big\},
% \end{equation}
% where $\Phi(\state^i, \tau)$ denotes the flow of the deterministic closed-loop dynamics of \eqref{eq:agent-dynamics} under control $\control^i = k(\state^i)$, mapping initial condition $\state$ at time $t$ to the state at time $t + \tau$, and $r > 0$ is the collision radius. We also introduce a critical time $\tau_c > 0$ as a lower bound on the inter-agent TTCs $\tau_{ij}^*$. 
% Practically, $\tau_c$ quantifies the minimum time available for a backup policy (e.g., emergency braking) to intervene before a collision occurs.

We are now ready to introduce one of our main contributions, a Stochastic Time-to-Collision CBF.
\begin{defn}\label{def:ttc-cbf}
    Given two agents $i,j \in \mathcal{I}$, the TTC metric \eqref{eq:time-to-collision}, and a critical time $\tau_c > 0$, the \textbf{Stochastic Time-to-Collision CBF} (S-TTC-CBF) is defined as
    \begin{equation}\label{eq:ttc-cbf}
        % B_{ij}(\state) = \frac{\tau_c}{\tau_{ij}^*(\state)},
        B_{ij}(\state) = \exp(\tau_c - \tau_{ij}^*(\state)),
    \end{equation}
    with the associated safe (collision-free) set
    \begin{equation}\label{eq:ttc-set}
        \Sset^{ij} = \left\{\state \in \Xset \mid \tau_{ij}^*(\state) > \tau_c\right\}.
    \end{equation}
\end{defn}
The exponential form in~\eqref{eq:ttc-cbf} is deliberate: by normalizing against $\tau_c$, the function $B_{ij}$ takes values in $(0, 1)$ on the set $\Sset^{ij}$, making it directly compatible with the unit-interval requirements of stochastic CBF theory.

For each pair $(i,j)$, define the set of admissible controls satisfying the SCBF condition:
\begin{equation}
    \mathcal{V}_{ij}(\state) \triangleq \{\control \in \Uset \mid \mathcal{A}B_{ij}(\state, \control) \leq -\alpha B_{ij}(\state) + \beta\}.
\end{equation}
The following theorem establishes the probabilistic safety guarantee provided by S-TTC-CBFs.
% \begin{thm}\label{thm:probabilistic-safety}
%     Consider the stochastic multi-agent system \eqref{eq:joint-dynamics} with S-TTC-CBFs $B_{ij}$ defined by \eqref{eq:ttc-cbf} and associated safe sets $\Sset^{ij}$ defined by \eqref{eq:ttc-set} for all pairs $i,j \in \mathcal{A}$. If a Lipschitz continuous control law satisfies $\control(\state) \in \bigcap_{i,j}\mathcal{V}_{ij}(\state)$, then the trajectories of \eqref{eq:joint-dynamics} are collision-free with probability at least $1 - \rho$, where $\rho$ is given by \eqref{eq:rho}.
% \end{thm}
\begin{thm}\label{thm:probabilistic-safety}
    Consider the stochastic multi-agent system~\eqref{eq:joint-dynamics} with S-TTC-CBFs $B_{ij}$ defined for all pairs $(i,j) \in \mathcal{P} = \{(i,j) : i,j \in \mathcal{A},\; i < j\}$ by~\eqref{eq:ttc-cbf} for sets $\Sset^{ij}$ defined by~\eqref{eq:ttc-set}. 
    Suppose that each $B_{ij}(\state_0) \leq \gamma_{ij} < 1$ with SCBF parameters $\alpha_{ij}, \beta_{ij} \geq 0$. 
    If a Lipschitz continuous control law satisfies $k(\state) \in \bigcap_{(i,j) \in \mathcal{P}} \mathcal{V}_{ij}(\state)$, then the trajectories of~\eqref{eq:joint-dynamics} are collision-free with probability at least $1 - \sum_{(i,j) \in \mathcal{P}} \rho_{ij}$, where each $\rho_{ij}$ is given by~\eqref{eq:scbf-safety-probability} with $(\alpha, \beta, \gamma) = (\alpha_{ij}, \beta_{ij}, \gamma_{ij})$.
\end{thm}
\begin{proof}
    % Let $\mathcal{P} = \{(i,j) : i,j \in \mathcal{A},\; i < j\}$ denote the set of all distinct agent pairs, with $|\mathcal{P}| = \binom{N}{2}$.
    Fix a pair $(i,j) \in \mathcal{P}$. 
    By hypothesis, $B_{ij}$ defined by~\eqref{eq:ttc-cbf} is an SCBF for \eqref{eq:joint-dynamics} with parameters $\alpha_{ij}, \beta_{ij} \geq 0$ and initial condition satisfying $B_{ij}(\state_0) \leq \gamma_{ij} < 1$.
    Since $k(\state) \in \bigcap_{i,j} \mathcal{V}_{ij}(\state)$, the SCBF condition~\eqref{eq:cbf-condition} holds for $B_{ij}$ along the closed-loop trajectories of~\eqref{eq:joint-dynamics}.
    Applying Theorem~\ref{thm:scbf-safety} to the pair $(i,j)$ yields
    \begin{equation}
        \mathbb{P}\!\left\{\exists\, t \in [0,T] :
        B_{ij}(\state_t) \geq 1\right\} \leq \rho_{ij},
    \end{equation}
    where $\rho_{ij}$ is given by~\eqref{eq:scbf-safety-probability} with $(\alpha, \beta, \gamma) = (\alpha_{ij}, \beta_{ij}, \gamma_{ij})$.

    A collision in the multi-agent system occurs if and only if at least one pair enters its unsafe set, i.e.,
    \begin{equation*}
    \begin{aligned}
        \big\{\exists\, t \in [0,T] : \state_t &\notin \textstyle\bigcap_{(i,j) \in \mathcal{P}} \Sset^{ij}\big\}
        = \\
        &\bigcup_{(i,j) \in \mathcal{P}} \big\{\exists\, t \in [0,T] :B_{ij}(\state_t) \geq 1\big\}.
    \end{aligned}
    \end{equation*}
    By the union bound,
    \begin{equation*}
        \mathbb{P}\!\left\{\exists\, t \in [0,T] : \state_t \notin
        \textstyle\bigcap_{(i,j) \in \mathcal{P}} \Sset^{ij}\right\}
        \leq \sum_{(i,j) \in \mathcal{P}} \rho_{ij}.
    \end{equation*}
\end{proof}
Note that when the SCBF parameters are uniform across all pairs ($\alpha_{ij} = \alpha$, $\beta_{ij} = \beta$, $\gamma_{ij} = \gamma$ for all $(i,j)$), the above simplifies to $\rho_{ij} = \rho$ for all pairs and the collision probability is $\sum_{(i,j) \in \mathcal{P}} \rho_{ij} = \binom{N}{2}\rho$.

By considering time-to-collision (TTC) between agents, one can account for the time required to execute a precautionary, evasive, or emergency maneuver (commonly referred to as a backup policy), rather than just their instantaneous separation and its rate of change. 
This approach naturally incorporates the dynamics of the system and ensures that agents have sufficient time to react and avoid collisions. 

We now define two associated sets:
\begin{align}
    \Sset_0^{ij} &= \{\state \in \Xset \mid \tau_{ij}^*(\state) > 0\}, \label{eq:positive-ttc} \\
    \Sset_B^{ij} &= \{\state \in \Xset \mid \tau_{ij}^*(\state) > 2\tau_c\}, \label{eq:backup-set}
\end{align}
comprising, respectively, the states with positive TTC and those with TTC exceeding $2\tau_c$.
\begin{assumption}\label{ass:backup}
    For each agent $i \in \mathcal{I}$, there exists a backup policy $\pi_B^i\colon \Xset^i \to \Uset^i$ such that, from any state $\state$ with $\tau_{ij}^*(\state) \leq \tau_c$ for some $j \neq i$, the closed-loop trajectory of agent $i$ under $\pi_B^i$ satisfies:
    \begin{enumerate}
        \item \textbf{(Recovery)} For the critical pair $(i,j)$,
        \begin{equation}
        \begin{aligned}
            \mathbb{P}\!\Big\{\Phi(\state^i, t{+}\tau) &\in \Sset_0^{ij},\; \forall \tau \in [0, \tau_c],\;\text{and} \\
            &\Phi(\state^i, t{+}\tau_c) \in \Sset_B^{ij}\Big\} \geq p_B,
        \end{aligned}
        \end{equation}
        \item \textbf{(Non-interference)} For all other pairs $(i,\ell)$ with $\ell \neq j$, the backup policy does not violate the SCBF condition, i.e.,
        \begin{equation}
            \mathcal{A}B_{i\ell}(\state, \pi_B^i(\state^i), \control^\ell) \leq -\alpha_{i\ell} B_{i\ell}(\state) + \beta_{i\ell},
        \end{equation}
        for all $\control^\ell \in \Uset^\ell$,
    \end{enumerate}
    for some $p_B \in (0,1)$.
\end{assumption}
% \begin{assumption}\label{ass:backup}
%     For each agent $i \in \mathcal{I}$, there exists a backup policy $\pi_B^i\colon \Xset^i \to \Uset^i$ such that, from any state $\state^i$ with $\tau_{ij}^*(\state) \geq \tau_c$ for any $j \neq i$, the closed-loop trajectory under $\pi_B^i$ satisfies
%     \begin{equation}
%     \begin{aligned}
%         \mathbb{P}\!\Big\{\Phi(\state^i, t+&\tau) \in \Sset_0^{ij},\; \forall \tau \in [0, \tau_c],\;\text{and} \\
%         &\Phi(\state^i, t{+}\tau_c) \in \Sset_B^{ij}\Big\} \geq p_B,
%     \end{aligned}
%     \end{equation}
%     for some $p_B \in (0,1)$.
% \end{assumption}
In short, this assumes that, with probability at least $p_B$, the backup policy keeps the system collision-free over the interval $[t, t+\tau_c]$ and returns it to the set $\Sset_B^{ij}$ by time $t + \tau_c$.
\begin{corollary}
    Suppose that Assumption~\ref{ass:backup} holds in addition to the premises of Theorem \ref{thm:probabilistic-safety}. 
    For each agent $i \in \mathcal{I}$, let $m^i(t) \in \{0, 1\}$ with $m^i(0) = 0$ and dynamics:
    \begin{equation}
        m^i(t^+) = \begin{cases}
            1, & \text{if } \, \exists j \neq i: \tau_{ij}^*(\state) \leq \tau_c, \\
            0, & \text{if } \, \tau_{ij}^*(\state) \geq 2\tau_c, \, \forall j \neq i \\
            m^i(t), & \text{otherwise.}
        \end{cases}
    \end{equation}
    Under the control policy where each agent $i \in \mathcal{I}$ applies
    \begin{equation}
        k^i(\state) = \begin{cases}
            u^i(\state), & \text{if } m^i(t) = 0, \\
            \pi_B^i(\state), & \text{if } m^i(t) = 1,
        \end{cases}
    \end{equation}
    where the joint nominal control $\control(\state) = (u^1(\state), \dots, u^N(\state))$ satisfies $\control(\state) \in \bigcap_{(i,j) \in \mathcal{P}} \mathcal{V}_{ij}(\state)$,
    the trajectories of~\eqref{eq:joint-dynamics} are collision-free with probability at least $1 - (1 - p_B)\sum_{(i,j) \in \mathcal{P}} \rho_{ij}$.

    % Under the following control policy,
    % \begin{equation}
    %     k(\state) = \begin{cases}
    %         \control(\state) \in \mathcal{V}_{ij}(\state), & \text{if } m(t) = 0, \\
    %         \pi_B(\state), & \text{if } m(t) = 1,
    %     \end{cases}
    % \end{equation}
    % the trajectories of \eqref{eq:joint-dynamics} are collision-free with probability at least $1 - (1 - \binom{N}{2}p_B)\sum_{(i,j) \in \mathcal{P}} \rho_{ij}$.
\end{corollary}
% \begin{proof}
%     The proof follows directly from Theorem \ref{thm:probabilistic-safety}, Assumption \ref{ass:backup}, and the independence of events $\omega_1 = \{\tau_{ij}^*(\state) \leq \tau_c\}$ and $\omega_2 = \{\Phi(\state^i, \pi_B^i, t+\tau) \in \Sset_0,\; \forall \tau \in [0, \tau_c],\;\text{and}
%         \Phi(\state^i, \pi_B^i, t{+}\tau_c) \in \Sset_B\}$.
% \end{proof}
\begin{proof}
    For each pair $(i,j) \in \mathcal{P}$, define the events
    \begin{align*}
        E_1^{ij} &\triangleq \left\{\exists\, t \in [0,T] : \tau_{ij}^*(\state_t) \leq \tau_c\right\}, \\
        E_2^{ij} &\triangleq \big\{\exists\, t \in [0,T] : m^i(t^-) = 0, m^i(t^+) = 1, \\
              &\quad\quad \{\Phi(\state_t, t+\tau_c) \notin \Sset_B^{ij} \;\text{or}\; \nonumber \\
              &\quad\quad \exists \tau \in [0, \tau_c]: \Phi(\state_t, t{+}\tau) \notin \Sset_0^{ij}\}\big\}, \nonumber \\
        E_C^{ij} &\triangleq \left\{\exists\, t \in [0,T] : \tau_{ij}^*(\state_t) \leq 0\right\}.
    \end{align*}
    That is, $E_1^{ij}$ denotes the CBF condition failing to maintain $\tau_{ij}^* > \tau_c$, $E_2^{ij}$ denotes backup activation followed by failure to recover to $\Sset_B^{ij}$ while remaining collision-free, and $E_C^{ij}$ denotes a collision.

    A collision between agents $i$ and $j$ requires both a breach of the CBF margin and a subsequent backup failure, so $E_C^{ij} \subseteq E_1^{ij} \cap E_2^{ij}$. Therefore,
    \begin{equation*}
        \mathbb{P}[E_C^{ij}] \leq \mathbb{P}[E_2^{ij} \mid E_1^{ij}]\,\mathbb{P}[E_1^{ij}].
    \end{equation*}
    When $m^i = 0$ for all $i \in \mathcal{I}$, the joint control satisfies $\control(\state) \in \bigcap_{(i,j) \in \mathcal{P}} \mathcal{V}_{ij}(\state)$, so the hypotheses of Theorem~\ref{thm:probabilistic-safety} hold and $\mathbb{P}[E_1^{ij}] \leq \rho_{ij}$.
    Conditioned on $E_1^{ij}$, when both agents $i$ and $j$ activate their backup policies the pair recovers successfully with probability at least $p_B$ by Assumption~\ref{ass:backup}, so $\mathbb{P}[E_2^{ij} \mid E_1^{ij}] \leq 1 - p_B$. Hence,
    \begin{equation*}
        \mathbb{P}[E_C^{ij}] \leq (1 - p_B)\,\rho_{ij}.
    \end{equation*}
    Any collision in the system requires at least one pairwise collision, so by the union bound,
    \begin{equation*}
        \mathbb{P}\!\left[\bigcup_{(i,j) \in \mathcal{P}} E_C^{ij}\right] \leq (1 - p_B) \sum_{(i,j) \in \mathcal{P}} \rho_{ij},
    \end{equation*}
    yielding collision-free probability at least $1 - (1 - p_B)\sum_{(i,j) \in \mathcal{P}} \rho_{ij}$.
    % A collision requires both a breach of the CBF margin and a subsequent backup failure, so $E_C^{ij} \subseteq E_1^{ij} \cap E_2^{ij}$. Therefore,
    % \begin{equation*}
    %     \mathbb{P}[E_C^{ij}]
    %     \leq \mathbb{P}[E_1^{ij} \cap E_2^{ij}]
    %     = \mathbb{P}[E_2^{ij} \mid E_1^{ij}]\,\mathbb{P}[E_1^{ij}].
    % \end{equation*}
    % By Theorem~\ref{thm:probabilistic-safety}, $\mathbb{P}[E_1] \leq \rho$.
    % Conditioned on $E_1$, the backup policy $\pi_B$ activates from a state with $\tau_{ij}^* \leq \tau_c$; by Assumption~\ref{ass:backup}, $\pi_B$ returns the system to $\Sset_B$ while remaining in $\Sset_0$ with probability at least $p_B$, so $\mathbb{P}[E_2 \mid E_1] \leq 1 - p_B$. Hence,
    % \begin{equation*}
    %     \mathbb{P}[E_C] \leq \rho(1 - p_B),
    % \end{equation*}
    % yielding collision-free probability $\mathbb{P}[E_C^c] \geq 1 - \rho(1 - p_B)$.
\end{proof}
Note that when $\tau_{ij}^*(\state) \leq \tau_c$, both agents $i$ and $j$ have $m^i(t^+) = m^j(t^+) = 1$ and independently switch to their respective backup policies $\pi_B^i$ and $\pi_B^j$.
In practice, backup policies depend on the system dynamics and often require pre-coordinated agreement in the multi-agent setting. 
The theoretical deconfliction probability achieved by a backup policy would be evaluated offline via, e.g., Monte Carlo analysis.

Thus far, the control policies governing the closed-loop flow $\Phi(\state, t)$ have been left unspecified, intentionally so as to provide a modular framework. We now fix a specific policy pair for defining a worst-case collision metric.

\subsection{Adversarial Time-to-Collision}
We build on the concept of TTC and define the \textit{adversarial time-to-collision} (aTTC) between agents $i, j \in \mathcal{I}$ as the earliest time at which the distance between $i$ and $j$ falls below the collision threshold $r$ under deterministic dynamics and the following assumptions: agent $i$ applies zero control ($\control^i_t = \mathbf{0}$) for all $t \in [0, T]$, while agent $j$ pursues $i$ at maximum acceleration and velocity. 
The aTTC thus captures the worst-case collision time between a nonreactive agent $i$ and a maximally adversarial agent $j$.
We ignore stochasticity when defining aTTC in order to arrive at a deterministic metric whose gradients are well defined and amenable to integration into a QP-based control law.
% \begin{defn}[aTTC]
%     Let each agent occupy a spherical region of space in $\R^3$ of radius $r$. The adversarial time-to-collision (aTTC) between agents i and j is then defined as 
%     \begin{equation}\label{eq:time-to-collision}
%         \tau^*  = \inf\big\{\tau \geq 0 \mid \|\Phi(\state^i, \mathbf{0}, \tau) - \Phi(\state^j, \pi_P^j(\state^j, \state^i), \tau)\| \leq 2r\big\},
%     \end{equation}
%     where $\pi_P^j$ is the maximally pursuant control policy.
% \end{defn}
\begin{definition}[aTTC]\label{def:attc}
    Given agents $i, j \in \mathcal{I}$ with $\state^i \in \Xset^i$, $\state^j \in \Xset^j$, the \textbf{adversarial time-to-collision} (aTTC) is
    \begin{equation}\label{eq:attc}
        \scalebox{0.95}{$\displaystyle
            \tau_{ij}^* = \inf\!\Big\{\tau \geq 0 \;\Big|\; \big\|\Phi^{\mathbf{0}}(\state^i, \tau) - \Phi^{\pi_P^j}\!\big(\state^j, \tau\big)\big\| \leq 2r\Big\},
        $}
    \end{equation}
    where $r > 0$ is the radius of each agent and $\pi_P^j\colon \Xset^j \times \Xset^i \to \Uset^j$ is the pursuit policy that maximizes the closing rate subject to velocity and acceleration bounds of agent $j$.
\end{definition}

%% Existance Proof
% \begin{lemma}\label{lemma:lem1}
%     If the trajectory path $S$ from $s_1$ to $s_2$ constitutes a feasible path for an agent moving at speed $v$, then  $S$ is a feasible path at any speed $v'<v$.
% \end{lemma}
% \begin{lemma}\label{lem:existance}
% For any \textit{controllable} dynamics (\ref{eq:agent-dynamics}) and agents $i,j$ with $v_{max,j}>v_{i}$ there exists at least one finite-time trajectory leading to collision and thus there exists a finite aTTC.  
% \end{lemma}
\begin{proposition}\label{lem:existence}
    Consider agents $i, j \in \mathcal{I}$ governed by~\eqref{eq:agent-dynamics} with agent $i$ applying zero control. If the maximum speed of agent $j$ exceeds that of agent $i$, i.e., $v_{\max}^j > v_{\max}^i$, and the dynamics of agent $j$ are controllable, then $\tau_{ij}^* < \infty$ for any initial configuration $(\state^i, \state^j) \in \Xset^i \times \Xset^j$.
\end{proposition}
% \renewcommand\qedsymbol{$\blacksquare$}
% \begin{proof}
% Consider a reference moving with velocity $v_{i}$. In this reference frame the evader (agent $i$), who is assumed to move ballistically, is stationary at point $x'_{i}$, and under the assumption that $v_{max,j}>v_{i}$ controllability implies $x'_{i}$ is reachable by the pursuer (agent $j$).
% \end{proof}
\begin{proof}
    Under zero control, agent $i$ follows the drift dynamics $\dot{\state}^i = f^i(\state^i)$, which, given the bound $v_{\max}^i$, traces a trajectory in $\Xset^i$ with a speed of at most $v_{\max}^i$. 
    Consider the relative displacement $\state^r(t) = \state^j(t) - \Phi(\state^i, \mathbf{0}, t)$. 
    In these coordinates, agent $j$ must drive $\|\state^r\|$ below $2r$. 
    Since $v_{\max}^j > v_{\max}^i$, agent $j$ has a strictly positive surplus speed $\Delta v = v_{\max}^j - v_{\max}^i > 0$. 
    By controllability, there exists a finite reorientation time $t_r \geq 0$ after which agent $j$ can direct its surplus velocity toward agent $i$, reducing $\|\state^r(t)\|$ at a rate of at least $\Delta v$ for all $t > t_r$. 
    During $[0, t_r]$, the separation grows by at most $(v_{\max}^i + v_{\max}^j)\, t_r$. 
    Therefore,
    \begin{equation}
        \tau_{ij}^* \leq t_r +
        \frac{\|\state^r(0)\|
        + (v_{\max}^i + v_{\max}^j)\, t_r - 2r}{\Delta v}
        < \infty,
    \end{equation}
    which establishes the existence of a finite aTTC.
\end{proof}

\begin{rem}
The proposition guarantees that the aTTC is well-defined and finite whenever a pursuer is faster than a nonreactive agent. Conceptually, aTTC quantifies the earliest time an agent would collide with the nonreactive agent under worst-case dynamics, not the agent's actual intent. 
% This distinction is critical and the framework provides robust safety guarantees for the evader, ensuring collision avoidance even under the most extreme physically feasible scenarios. 
As such, aTTC serves as a robust metric of adversarial risk.
% , forming the foundation for predictive control strategies that \textcolor{red}{minimizes collision risk} under all admissible uncertainties.
\end{rem}

\subsection{Surrogate Modeling}\label{sec:ml}
Direct computation of the aTTC requires integrating the system dynamics forward from the current state to predicted collision, a computationally demanding task in all but the simplest models. Furthermore, integration into an optimization based control law requires a differentiable representation of the mapping from the system state to the aTTC. In order to build a fast and differentiable surrogate model we use a neural network. 

For any two agents, the aTTC depends primarily on their relative position $\Delta x \equiv x_i-x_j$, and their velocities as $v_i$ $v_j$. Secondarily to the relative states, the aTTC also depends on three parameters: the collision radius $r$, the maximum velocity: $v_{max}$, and the maximum acceleration $a_{max}$. For clarity of exposition we will assume constant $r$ and $ a_{max}$. We then seek a parametric model
\begin{equation}\label{eq:ml_model}
    G_{\theta}[\Delta x,v_i,v_j,v_{max}]:\Xset \times \R\times \R\times \R \to \R
\end{equation}
with parameters $\theta$ which models the forward integration of (\ref{eq:agent-dynamics}) implicit in the definition (\ref{eq:time-to-collision}).  As our surrogate model we utilize a simple 3-layer fully connected neural network (NN) (layer sizes: 128-128-64) which maps $\Delta x,v_i,v_j, v_{max} \rightarrow \tau^*$. As the aTTC is a deterministic function of the system state at any instance in time no recurrent architecture is needed. As previously mentioned, the aTTC depends critically on not just the state but also on the assumed system parameters, most notable the maximum speed of the pursuer. Therefore if the surrogate model is to be applicable to a range of possible pursuer types, its predictions should be \textit{conditioned} on those parameters. In this case, we incorporate the parametric dependence on $v_{max}$ through Feature-wise Linear Modulation (FiLM) \cite{perez_film_2017} which adds a small parallel fully-connected branch that takes as input the system parameters and modulates the activation function of the main branch. Specifically, for a layer with weights and biases $W$ and $b$, and activation function $\operatorname{a}(\cdot)$ the standard output $y_{out} = \operatorname{a}(Wx +b)$ is replaced by 
\begin{align*}
    z&=\gamma \odot (Wx +b) +\beta \\
    y_{out} &= \operatorname{a}(z)
\end{align*}
where $\gamma$ and $\beta$ are learned by a parallel branch of the NN and $\odot$ is the element-wise Hadamard product. 
% An illustration of the architecture is shown in figure \ref{fig:ml}.
The NN is trained using a weighted Huber loss
\begin{equation*}
    \mathcal{L}\left(\hat{\tau}^*,\tau^*\right)=\frac{1}{\tau^*}\mathcal{L}_H\left(\hat{\tau}^*,\tau^*\right)
\end{equation*}
where the $1/\tau^*$ weight emphasizes low aTTC (high risk) configurations. We utilize this simple weight as opposed to probability-based weighted losses \cite{rudy_output-weighted_2021}, as these overemphasize the rare, but unimportant, large values of $\tau^*$. The Huber loss itself is defined as
\begin{equation*}
\mathcal{L}_H\left(\hat{\tau}^*,\tau^*\right) =
    \begin{cases}
\frac{1}{2}(\tau^* - \hat{\tau}^*)^2, & \text{if } |\tau^* - \hat{\tau}^*| \leq \delta \\
\delta \left( |\tau^* - \hat{\tau}^*| - \frac{1}{2}\delta \right), & \text{otherwise}.
\end{cases}
\end{equation*}
The Huber loss \cite{huber_robust_1964} applies a mean-squared-error (MSE) loss to small outputs and a linear loss to larger values. This ensures smooth behavior near zero without overly weighting errors for large values of $\tau^*$.

%% Control Strategy
\subsection{Control Strategy}\label{sec:cbf}
As a control policy, we adopt the following CBF-QP-based control law:
\begin{subequations}\label{eq:cbf_qp_controller}
\begin{align}
    \control^*(x) &= \argmin_{\control \in \Uset}\;\frac{1}{2}\|\control - \control_0(\state)\|_{M}, \label{subeq:objective} \\
    \;\;& \textrm{s.t.} \nonumber \\
    \mathcal{A}B(\state, &\control) \leq -\alpha B(\state) + \beta, \label{subeq:scbf_constraint}
\end{align}
\end{subequations}
% \begin{equation}
%     \control^*(x) = \argmin_{\control \in \Uset}\;\frac{1}{2}\|\control - \control_0(\state)\|_{M} \; \text{s.t.} \; \mathcal{A}B_\theta(\state, \control) \leq -\alpha B_\theta(\state) + \beta
% \end{equation}
where \eqref{subeq:objective} seeks to minimize the deviation from the nominal controller $u_0: \Xset \to \mathcal{U}$ under a weighting matrix $M \in \R^{m \times m}$, and \eqref{subeq:scbf_constraint} enforces the stochastic CBF constraint. We define $B$ in terms of the learned surrogate model for the aTTC (\ref{eq:ml_model})
\begin{equation}\label{eq:barrier}
    B(x)=\exp\left(\tau_c - G_{\theta}[\Delta x,\Delta v,v_{max}]\right),
\end{equation}
where $\tau_{c}$ is a user defined safety threshold. Note that as a NN, $G_{\theta}$ is fully differentiable so the generator in (\ref{subeq:scbf_constraint}) may be computed at each time step. 
% \textcolor{red}{We note that (\ref{subeq:scbf_constraint}) guarantees forward invariance of $\Sset$ \textit{only} when assuming the dynamics (\ref{eq:agent-dynamics}) evolve deterministically. In the stochastic setting, the process noise $w_i$ may violate this condition regardless of the applied control $u_i$. Additionally, consistent with Lemma \ref{lem:existence} adversarial pursuit conditions may lead to configurations where collision is unavoidable.} 
We note that NNs have also been used to learn forward invariance conditions directly from data \cite{dawson_safe_2021,dawson_safe_2022,yu_sequential_2023,so_how_2023} an idea that has recently been extended to multi-agent systems through graph NNs \cite{zhang2025gcbf+}. Additionally, CBFs have been incorporated directly as differentiable layers into the architecture of NNs \cite{xiao_barriernet_2023}. These have the advantage of being compatible with any NN based controller and are trainable by existing gradient descent algorithms. 
%We note that our approach is 

\section{Simulation \& Results}\label{sec:results}
% \begin{table}[]
%     \centering
%     \begin{tabular}{c|c|c|c}
%        $<10$s  & $10-30$s & $30-60$s & $>60$ s   \\
%        \hline
%         1.13s & 2.47s & 8.33s & 22.8s
%     \end{tabular}
%     \caption{Mean absolute error of NN surrogate model applied to unseen data for different ranges of true $\tau^*$}
%     \label{tab:ml_details}
% \end{table}

\subsection{Comparative Study: Planar Collision Avoidance}

Our first case study examines the proposed aTTC-CBF-based controller on a 2D (planar) double integrator over 1000 stochastic trials. 
An aTTC-CBF-controlled agent perturbed by Brownian motion process noise travels along the x-axis in the positive direction, while a second, constant-velocity agent travels along the y-axis in the negative direction, creating a potential collision near the origin.
The first agent's diffusion term is taken to be $\sigma(\state) = \mathrm{diag}(1, 0, 0, 0)$, and it starts at $[x, y, v_x, v_y] = [-5, 0, 1, 0]$, whereas the second agent begins at $[0, 5, 0, -1]$. Under constant-velocity motion and an agent radius of $r=1$, this configuration leads to a collision at $t \approx 3.6$s.
The first agent's control is restricted to acceleration along the x-axis (i.e., no y-direction input), and the second agent applies zero control.
As described in Section \ref{sec:ttc-cbf}, we train a NN to generate the aTTC value at runtime.
For comparison, we also simulate the first agent under a collision cone CBF \cite{tayal2026collision}, a velocity obstacle CBF \cite{roncero2025multi}, and a relaxed future-focused CBF \cite{black2023future}. 
Each uses the control law \eqref{eq:cbf_qp_controller} with $B(\state) = \exp(-kh(\state))$, where $h$ is the form of (zeroing) CBF of the compared work, $k > 0$ is chosen so that $B(\state_0) = 0.1$, and parameters $\alpha = 1$, $\beta = 0.018$, yielding a collision risk bound of $\rho \leq 0.25$ for all controllers.

The results of the study are summarized in Table~\ref{tab:2d-results}. The aTTC-CBF controller achieves a safety rate of 0.918, well above the 0.75 guaranteed by Theorem~\ref{thm:scbf-safety}, while also attaining the highest goal rate (the frequency of successfully reaching the origin) at 63\%. It also records the second-longest average time to first collision, trailing only the VO-CBF method, which is so conservative that it crosses the origin in just 4\% of trials. 

Figure~\ref{fig:cbf-2d} shows the mean values of $B(\state)$ for each controller, with shaded regions indicating the $+1\sigma$ range. 
The aTTC-CBF maintains consistently low values of $B$ (safe), while the other methods exhibit spikes corresponding to dangerous configurations both early and late in the simulation.
Figure~\ref{fig:x-2d} shows the mean $\pm 1\sigma$ x-trajectories of the first agent under each controller.
The aTTC-CBF makes steady, collision-free progress toward its goal, whereas the remaining methods either are overly conservative (C3BF, VO-CBF) or collide at rates exceeding the theoretical bound (RFF-CBF), indicating that the S-CBF condition cannot be satisfied at all times.

In what follows, we examine the effect of aTTC-CBF control in a multi-agent, pursuit-evasion scenario.

\begin{figure}
    \centering
    \includegraphics[width=\columnwidth]{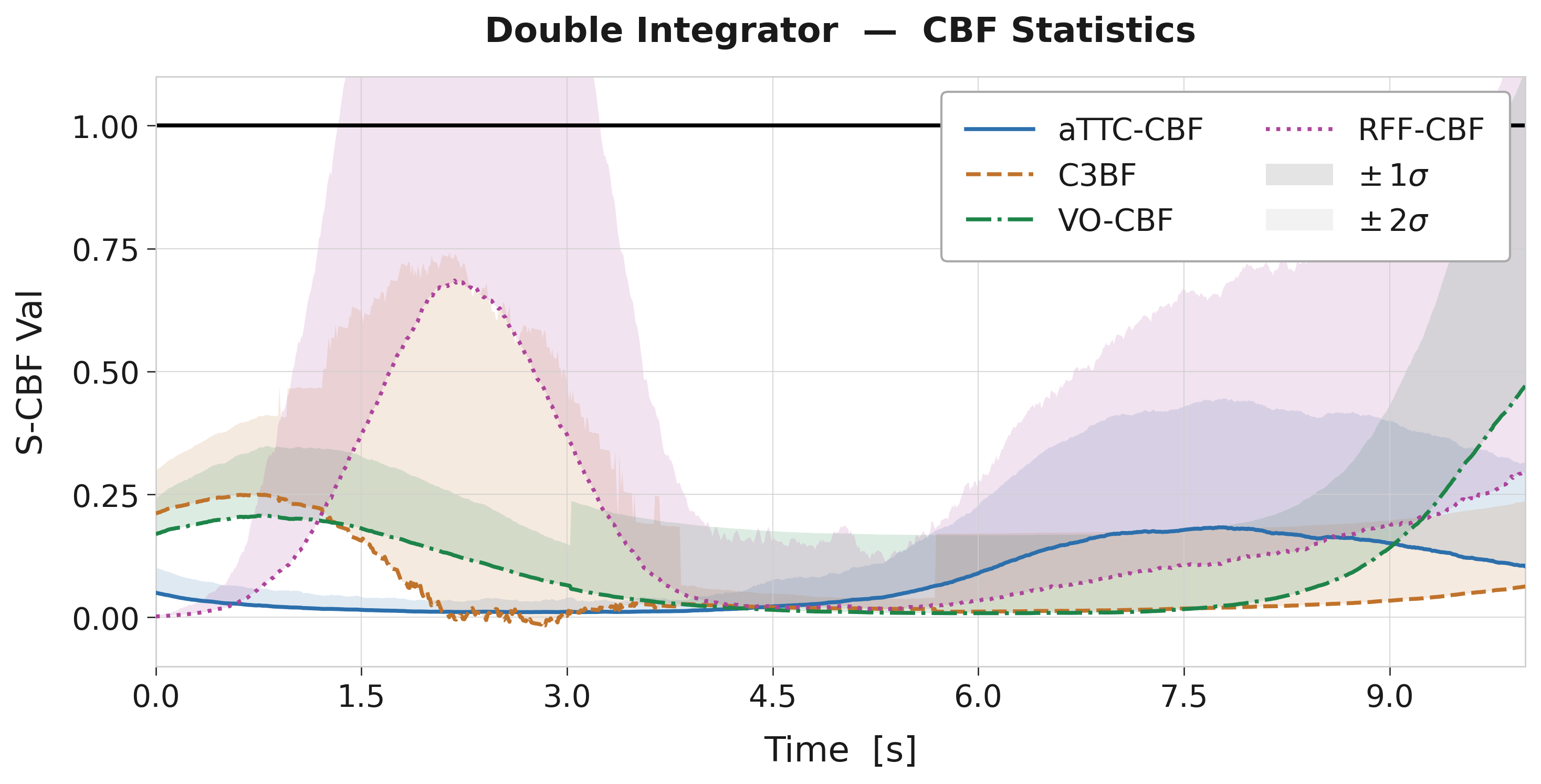}
    \caption{Mean and +1$\sigma$ S-CBF values for the four compared controllers in the planar collision avoidance example over 1000 simulated trials.}
    \label{fig:cbf-2d}
\end{figure}
\begin{figure}
    \centering
    \includegraphics[width=\columnwidth]{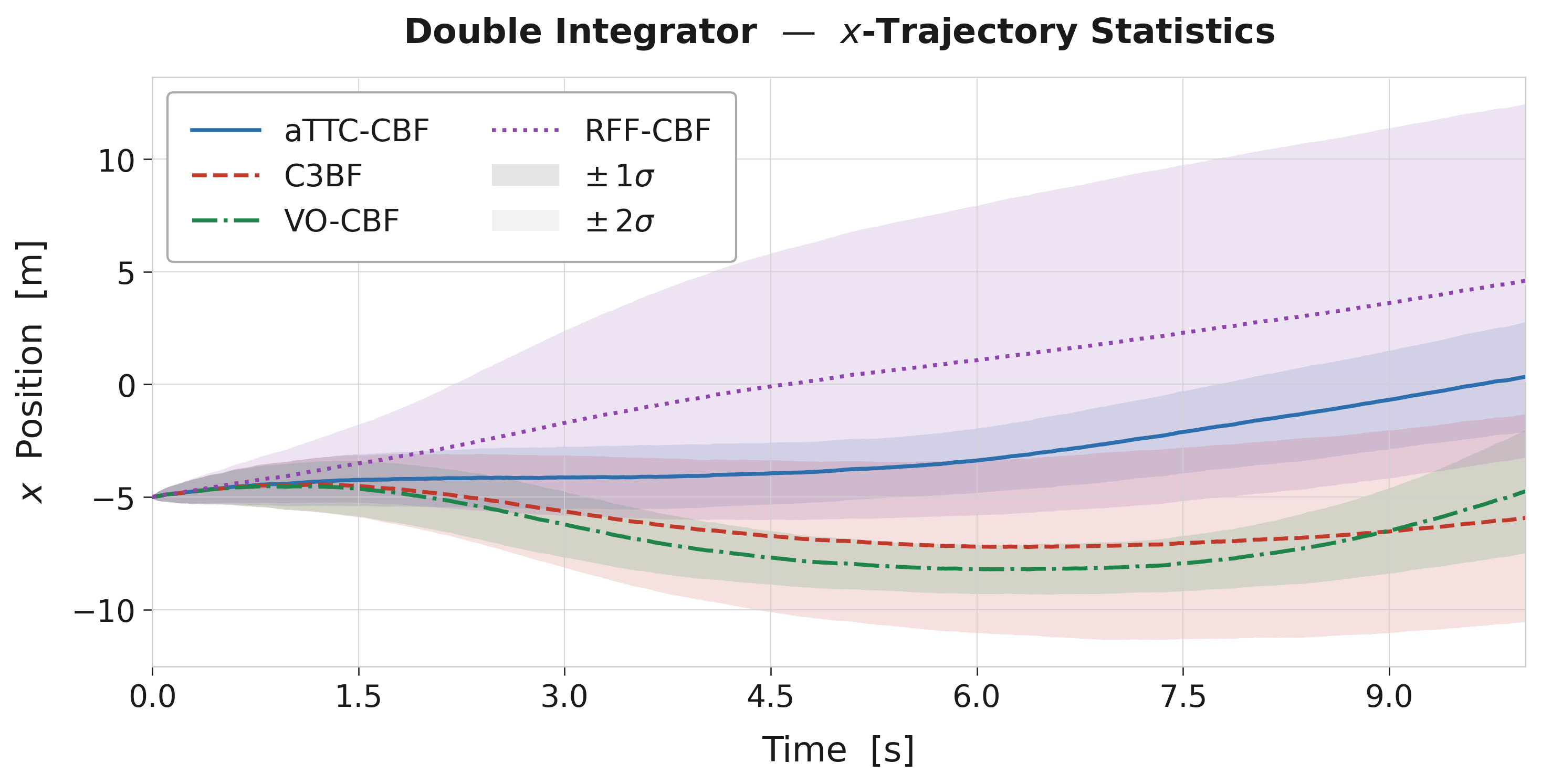}
    \caption{Mean and +/-1$\sigma$ $x$-coordinate trajectories for the four compared controllers in the planar collision avoidance example over 1000 simulated trials.}
    \label{fig:x-2d}
\end{figure}
\begin{table}[]
    \centering
    \scalebox{0.98}{%
    \begin{tabular}{c|c|c|c|c}
       & aTTC-CBF & C3BF & VO-CBF & RFF-CBF \\
       \hline
        Safety Rate & 0.918 & 0.976 & 0.783 & 0.624 \\
        \hline
        Goal Rate & 0.630 & 0.060 & 0.043 & 0.590 \\
        \hline
        Avg. 1$^{\mathrm{st}}$ Collision (s) & 7.80 & 1.49 & 9.38 & 3.62 
    \end{tabular}}
    \caption{Safety, goal progress, and collision timing statistics for the planar study.}
    \label{tab:2d-results}
    \vspace{-5mm}
\end{table}

\subsection{Comparative Study: 3D Dubin's Dynamics}
\subsubsection{Neural Network Training}
To generate the training data we construct a series of scenarios where two agents governed by the 3D Dubin's model described in~\cite{mclain2014implementing} navigate 3D space, e.g., follow randomly spaced way-points. The various scenarios are chosen to expose the NN to a wide array of relative positions and velocities. 
For these trajectories we then compute the aTTC for a range of $v_{max}$ each of which leads to an independent set of paired training data: $[\Delta x,v_i,v_j,v_{max}]\rightarrow \tau^*(v_{max})$.
The total training data consists of 50,000 seconds of simulation at time steps of 0.1 for which we compute the aTTC at three distinct pursuer-to-evader max speed ratios: 1.2, 2.4, and 3.6. The first 75\% of the data is used for training and the remaining 25\% is used to evaluate the trained model. Figure \ref{fig:nn_error} shows the conditional mean absolute error (MAE) and mean percentage error (MPE) -- that is to say the error when evaluated on samples with $\tau^* <T$ as a function of $T$. We note that tor high-risk configurations ($\tau^*<10$s) the model achieves an MAE of well under 1s and that the NN generalizes well, having comparable errors on both the training and test data sets.

\begin{figure}
    \centering
    \includegraphics[width=0.49\textwidth]{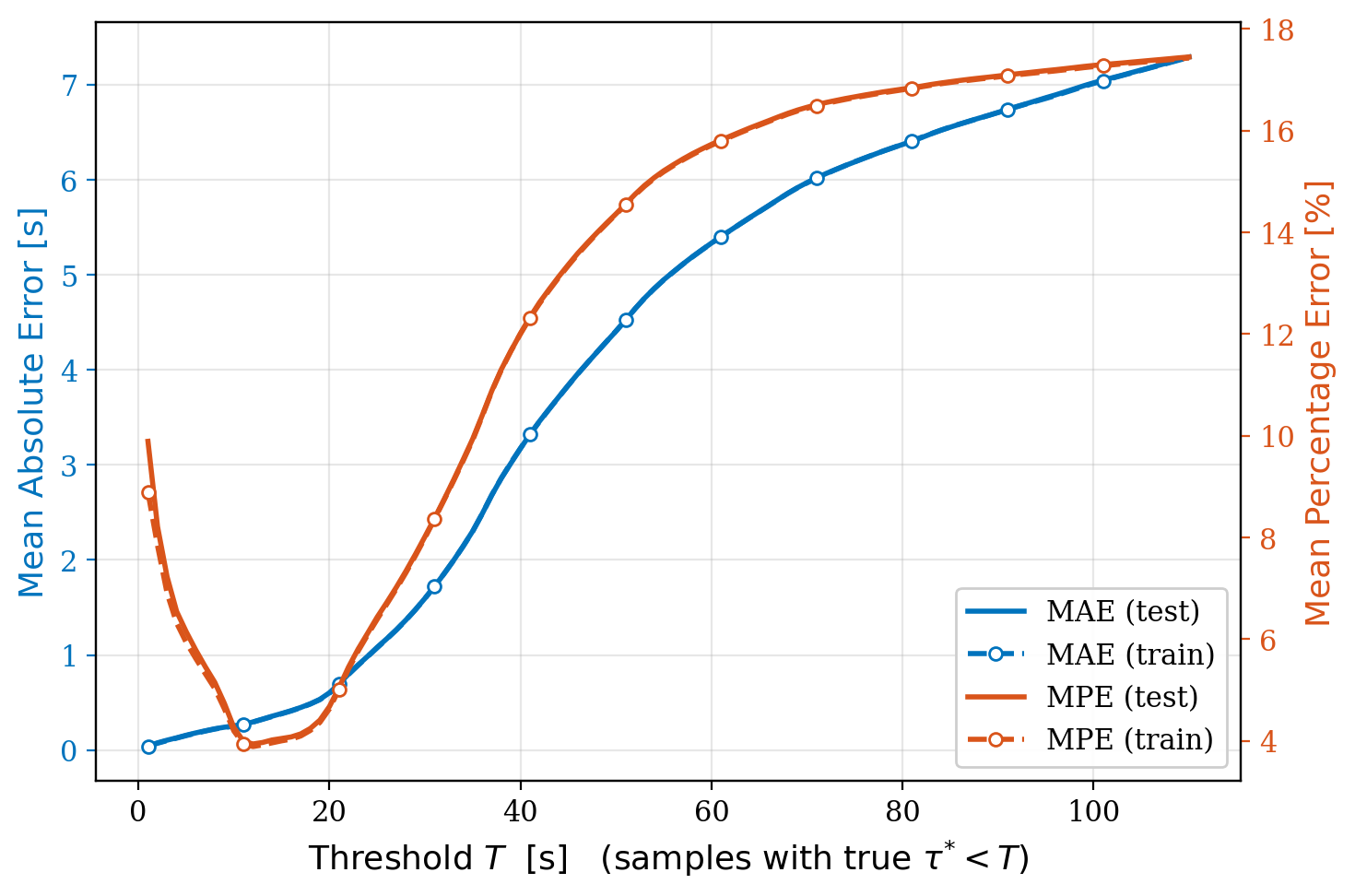}
    \caption{Conditional MAE and MPE for Dubins aTTC NN surrogate model. Training data in dashed lines w/circles, test data in solid lines.}
    \label{fig:nn_error}
\end{figure}
%  Additionally, we found that on average  the NN-model \textit{underestimates} the aTTC meaning that the surrogate model prediction is conservative -- a generally more acceptable form of bias for collision avoidance.
\subsubsection{Numerical Experiment}
As a numerical experiment we consider a stochastically forced 3D Dubin's model \cite{mclain2014implementing} where each agent's state $ \boldsymbol{x}= [x,y,z,\gamma,\psi,v]^T$ corresponds to the spatial positions $x,y,z$ as well as the pitch, and yaw angles $\gamma, \psi$, and speed $v$. The forcing (process noise) is assumed to be a Weiner process with fixed variance. 
All agents are equipped with proportional heading and speed control. 
% Inspired by the aerospace application relevant to this study, our simulation is conducted in nominal units of kilometers and seconds. 
There is an evader agent, whose nominal control objective is to orbit a fixed way-point at a desired speed of $0.25$km/s while avoiding collisions with five pursuer agents using the aTTC-based CBF introduced in Section~\ref{sec:cbf}. 
The pursuers attempt to intercept the evader through a proportional control scheme that navigates towards a far off point collinear with the pursuer and agent, which ensure the pursuer attempts to intercept the evader at maximum speed, which is set to $v_{max,p}=0.75$km/s (50\% higher than that of the evader $v_{max,e} = 0.5$km/s). In order to obtain a statistical view of the collision avoidance skill we run the simulation for $T_{sim} = 50,000$s.
% A summary of the additional simulation details are summarized in table \ref{tab:sim_details} and CBF the parameters and numerical overhead of each method is provided in table \ref{tab:cbf_details}. 

In order to illustrate the benefits of the aTTC-CBF we compare our method to a higher-order CBF (HOCBF)~\cite{xiao2021high}. To facilitate a fair comparison the CBF parameters were made as similar as possible. Specifically, the aTTC-CBF relies on $\alpha$ and $\beta$ in (\ref{subeq:objective}) as well as the critical time $\tau_{c}$ in (\ref{eq:barrier}). The HOCBF relies on three constants $\alpha_1,\alpha_2$,  $\beta$ and a critical distance $r_{min}$. We fix $\beta = 0.01, ~\ \alpha=\alpha_1=0.1, ~\  \alpha_2 = 1, ~\ \tau_{c}=5$s, and $r_{cbf} = 1$km. We note that for this example, both the aTTC-CBF and the HOCBF have a per-call overhead of less than 1.0 ms in Matlab on a CPU and that 
due to the challenging setup of this problem, the CBF constraint will almost always be violated and the evader is generally trying to \textit{get to safety} rather than \textit{stay safe}. We also emphasize that this comparison is \textit{not} intended to demonstrate the superiority of our approach to any \textit{specific} state-of-the-art method but rather to illustrate phenomenologically the advantages of TTC-based CBFs. 

The headline results are summarized in Table \ref{tab:dubins_results}, the aTTC-CBF succumbs to 3.7 collisions per 100s vs. 16.2 per 100s for the HOCBF -- more than 75\% less. The aTTC-CBF is able to achieve this by directly encoding the \textit{dynamics constrained} closing speed between it and its pursuers -- which allows it to more effectively modulate it's own velocity. As a proxy for the types of evasive maneuvers demanded by each CBF we show in Figure \ref{fig:col_scatter} a scatter plot of the horizontal and vertical velocity components $v\cos({\gamma})$ and $v\sin(\gamma)$ \textit{at} the time of collision. The aTTC-CBF equipped agents exhibit significantly more velocity modulation -- with collision speeds being relatively evenly distributed between 0 and the maximum. The HOCBF on the other hand uses very little vertical motion to escape the pursuers, instead relying purely on horizontal maneuvers. This discrepancy is directly a consequence of the dynamics-aware nature of aTTC. In this numerical experiment, the agents have a maximum yaw rate, $\dot{\psi}_{max}$ (horizontal turns) four-times greater than the maximum pitch rate, $\dot{\gamma}_{max}$ (vertical turns). The HOCBF is purely geometric and, knowing nothing of these limits, is biased towards acting in the direction with the highest control authority -- in this case yaw. The aTTC-CBF on the other hand has implicitly learned these control limits and knows that the pursuers are subject to the same $\dot{\gamma}_{max} < \dot{\psi}_{max}$ constraint which allows it to more effectively apply the appropriate evasive strategy.

This improved velocity modulation allows the aTTC-CBF to navigate between the pursuers more effectively. This is illustrated in Figure \ref{fig:col_dist} which shows the probability density function (PDF) of minimum distance between the evader and the nearest pursuer, defined as $\operatorname{min}_j(|x_1-x_j|)$. The aTTC-CBF spends less time in the collision zone but also spends more time at intermediate distances (0.25-1.25km), indicating navigation in-and-among the pursuers. 
This maneuverability means that the aTTC-CBF equipped agents more effectively escape close encounter -- defined here as a pursuer coming within twice the collision distance of the evader. For the HOCBF agents 69\% of close encounters result in collisions, where as the aTTC-CBF agents escape 50\% of these -- see Table \ref{tab:dubins_results}.

Finally, to evaluate the scalability of our framework we repeated the above experiment for increasing number of pursuing agents -- holding all other parameters constant. We found that as expected the computational cost and the total collision rate scale linearly with the number of agents, though tabulated data was omitted due to space constraints. All the code needed to reproduce these results as well as videos of the simulation can be found at \url{https://github.com/ben-barthel/Temporal_Barrier_Functions}

% In summary, the aTTC-CBF's advantage is that it understands closing velocity. When a faster pursuer is closing in, the aTTC-CBF recognizes that high relative speeds are dangerous even at moderate distances, and responds by decelerating, which buys time for heading corrections. 

% \begin{table}[]
%     \centering
%     \begin{tabular}{c|c|c|c|}
%        $N_{evader}$ & $N_{pursuer}$ & $T_{sim}$ & $v_{max,p}/v_{max,e}$   \\
%        \hline
%          1 & 5& 50,000 & 1.5 
%     \end{tabular}
%     \caption{Pursuit simulation details: Evader and pursuer count, simulation length, and pursuer speed advantage.  }
%     \label{tab:sim_details}
% \end{table}
% \begin{table}[]
%     \centering
%     \begin{tabular}{c|c|c|c|c|c}
%         $\alpha$ & $\tau_{min}$ & $k$ & $\alpha_1,\alpha_2$ & $r_{min}$  & $k$   \\
%        \hline
%         0.1 & 5.0 s & 0.5 & 0.1,1.0 & 1.0 km & 0.5
%     \end{tabular}
%     \caption{CBF parameters for aTTC-CBF: $(\alpha, \tau_{min},k_{attc})$ and HOCBF: $(\alpha_1,\alpha_2,r_{min}, k_{hocbf}$ )}.
%     \label{tab:cbf_details}
% \end{table}

\begin{table}[]
    \centering
    \begin{tabular}{c|c|c}
         & aTTC-CBF & HOCBF \\
       \hline
       Collision Rate (cols/100s)  & 3.7 & 16.2 \\
       \hline
       Collision / Close Encounter  & 0.50 & 0.69 
    \end{tabular}
    \caption{Headline results for multi-agent pursuit evasion. Total collision rate and fraction of close encounters: $d<4r$ which resulted in a collision: $d<2r$ per evader-pursuer pair.}
    \label{tab:dubins_results}
\end{table}

% \begin{figure}
%     \centering
%     \includegraphics[width=0.4\textwidth]{figures_cdc_final/collision_stats.png}
%     \caption{Total collision rate (a) and fraction of close encounters: $d<4r$ which resulted in a collision: $d<2r$ per evader-pursuer pair. aTTC-CBF (blue), HOCBF (red)}
%     \label{fig:col_rate}
% \end{figure}
\begin{figure}
    \centering
    \includegraphics[width=0.45\textwidth]{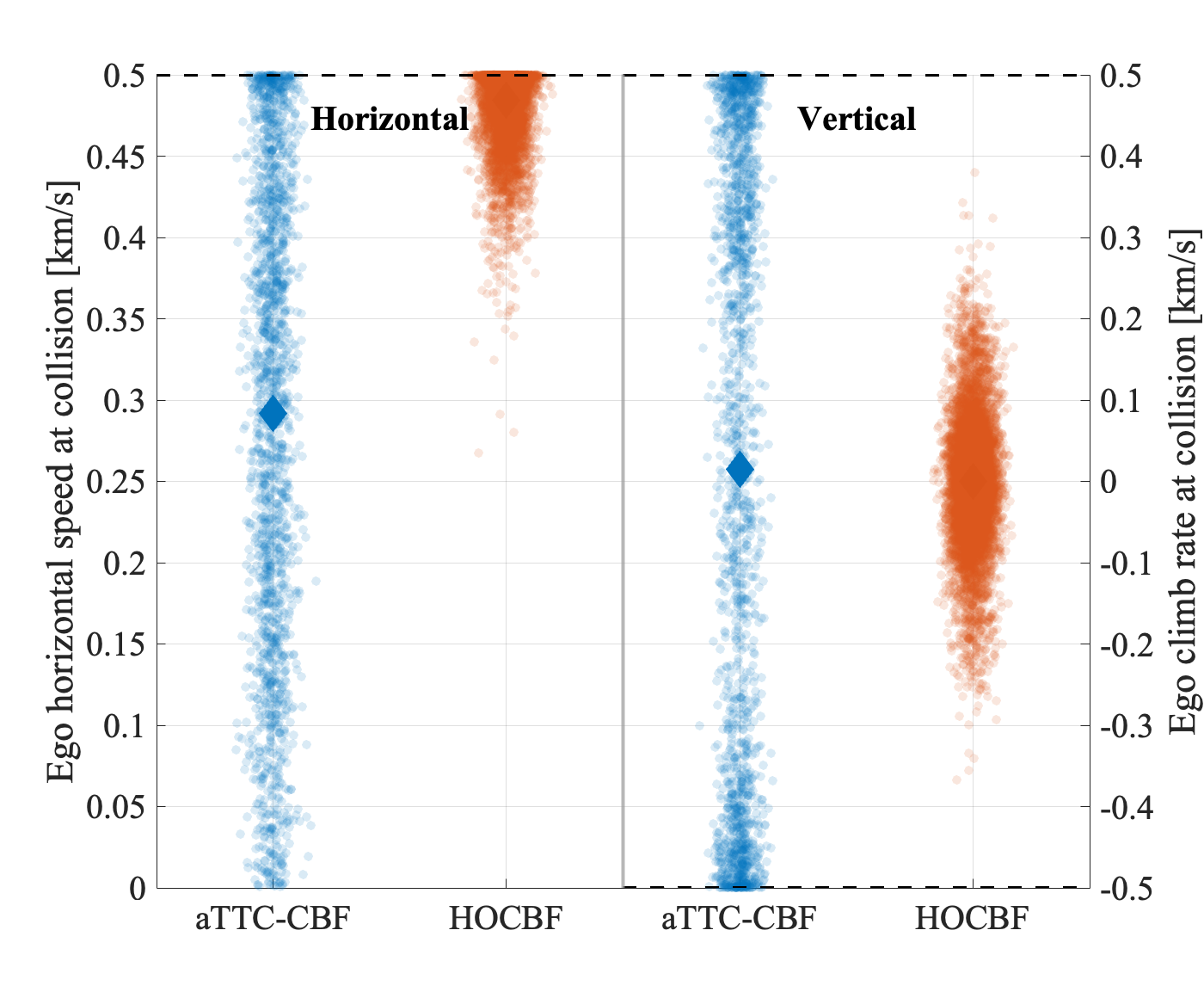}
    \caption{Scatter plot showing distribution of horizontal and vertical velocity at time of collision with mean values indicated with triangles -- horizontal spacing is purely to improve readability.}
    \label{fig:col_scatter}
\end{figure}

\begin{figure}
    \centering
    \includegraphics[width=0.45\textwidth]{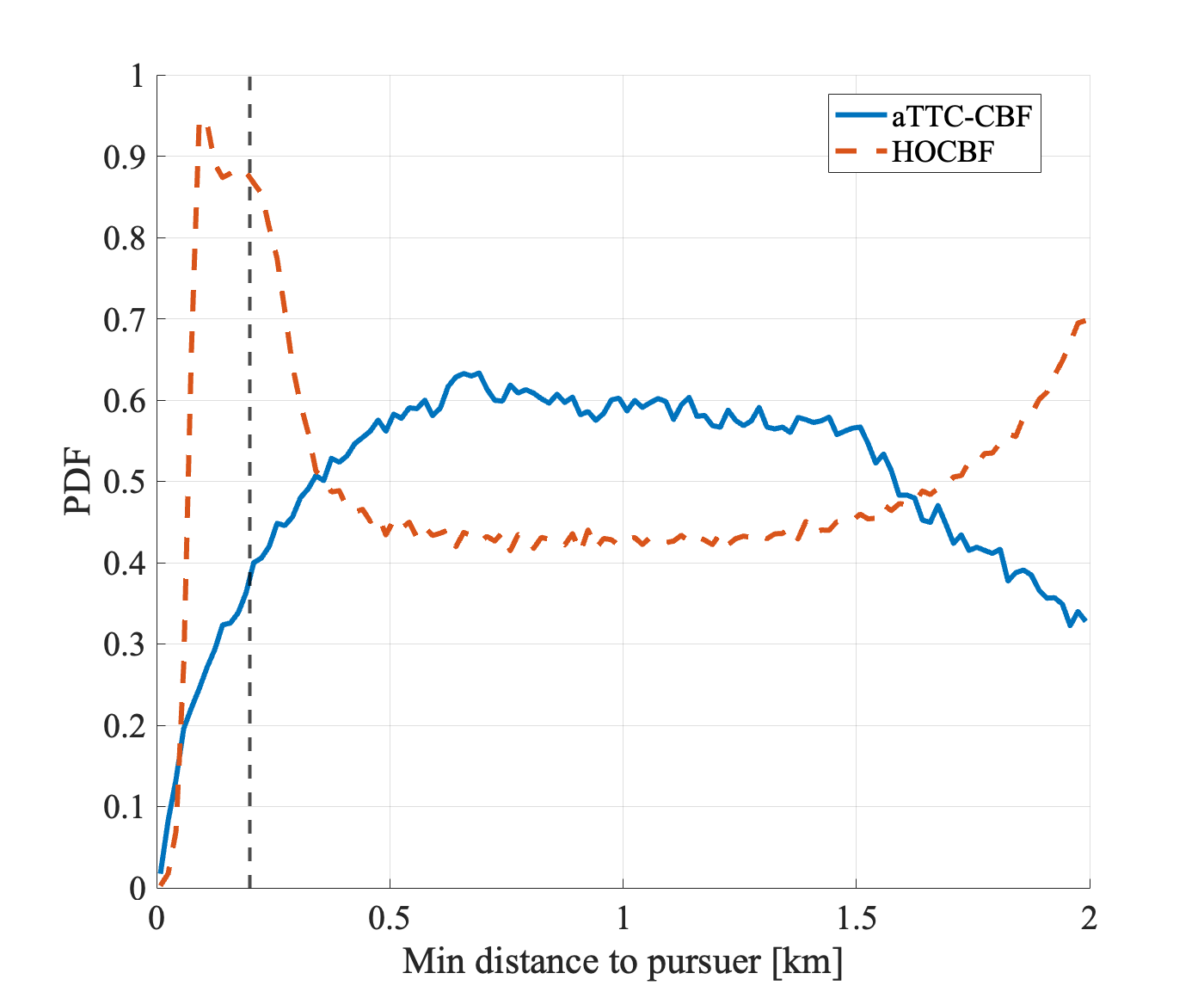}
    \caption{Probability density function (PDF) of minimum distance: $\operatorname{min}_j(|x_1-x_j|)$ from evader to nearest pursuer. aTTC-CBF (blue), HOCBF (red).}
    \label{fig:col_dist}
\end{figure}

\section{Conclusions}\label{sec:conclusion}
In this work, we introduced time-to-collision (TTC) as a basis for proactive CBF-based control in stochastic multi-agent systems. 
We established a sufficient condition for probabilistic collision avoidance and proposed an adversarial TTC variant to provide robustness in the multi-agent setting.
We demonstrated the success of our proposed controller on two case studies: a 2D double integrator benchmarked against existing methods, and a multi-agent pursuit-evasion scenario involving fixed-wing aircraft, both of which highlighted the advantages of using TTC in a CBF-based control law.

\section*{Acknowledgments}
This material is based upon work supported by the Department of the Air Force under Air Force Contract No. FA8702-15-D-0001 or FA8702-25-D-B002. Any opinions, findings, conclusions or recommendations expressed in this material are those of the author(s) and do not necessarily reflect the views of the Department of the Air Force. 
% © 2026 Massachusetts Institute of Technology.
% Delivered to the U.S. Government with Unlimited Rights, as defined in DFARS Part 252.227-7013 or 7014 (Feb 2014). Notwithstanding any copyright notice, U.S. Government rights in this work are defined by DFARS 252.227-7013 or DFARS 252.227-7014 as detailed above. Use of this work other than as specifically authorized by the U.S. Government may violate any copyrights that exist in this work.
%\input{sections-cdc-final/7_supplementary}

\bibliographystyle{IEEEtran}
\begingroup
\renewcommand{\url}[1]{}
\bibliography{refs/references, refs/references_2}
\endgroup

\end{document}